\documentclass[a4paper]{article}

\usepackage{amsmath}
    \addtocounter{MaxMatrixCols}{1} 

\usepackage{amsthm}
    \newtheorem{theorem}{Theorem}
    \newtheorem{proposition}[theorem]{Proposition}
    \newtheorem{corollary}[theorem]{Corollary}
    \newtheorem{remark}[theorem]{Remark}
    \newtheorem{lemma}[theorem]{Lemma}
    \newtheorem{example}[theorem]{Example}

\usepackage{graphicx,booktabs,enumitem,url,pstricks,pst-node}

\newcommand{\plex}{\prec_\mathrm{lex}}
\newcommand{\rem}{\ \mathrm{rem}\ }
\newcommand{\F}{\mathbf{F}}
\renewcommand{\tilde}[1]{\widetilde{#1}}

\begin{document}

\title{On Field Size and Success Probability in Network Coding}

\author{Olav Geil%
\thanks{Department of Mathematical Sciences,
Aalborg University, Denmark, \protect\url{olav@math.aau.dk}},
Ryutaroh Matsumoto%
\thanks{Department of Communications and Integrated Systems,
Tokyo Institute of Technology, Japan
\protect\url{ryutaroh@rmatsumoto.org}},
Casper Thomsen%
\thanks{Department of Mathematical Sciences,
Aalborg University, Denmark,
\protect\url{caspert@math.aau.dk}}}

\maketitle  

\begin{abstract}
    Using tools from algebraic geometry and Gr\"obner basis theory we solve two
    problems in network coding. First we present a method to determine the
    smallest field size for which linear network coding is feasible. Second we
    derive improved estimates on the success probability of random linear
    network coding. These estimates take into account which monomials occur in
    the support of the determinant of the product of Edmonds matrices. Therefore
    we finally investigate which monomials can occur in the determinant of the
    Edmonds matrix.\\

\end{abstract}

\section{Introduction}\label{secintro}

In a traditional data network, an intermediate node only forwards data and never
modifies them.  Ahlswede et~al.~\cite{ahlswede00} showed that if we allow
intermediate nodes to process their incoming data and output modified versions
of them then maximum throughput can increase, and they also showed that the
maximum throughput is given by the minimum of maxflows between the source node
and a sink node for single source multicast on an acyclic directional network.
Such processing is called network coding.  Li et~al.~\cite{li03} showed that
computation of linear combinations over a finite field by intermediate nodes is
enough for achieving the maximum throughput.  Network coding only involving
linear combinations is called linear network coding.  The acyclic assumption
was later removed by Koetter and M\'edard \cite{koetter03}.

In this paper we shall concentrate on the error-free, delay-free multisource
multicast network connection problem where the sources are uncorrelated.
However, the proposed methods described can be generalized to deal with delays
as in \cite{ho06}. The only exception is the description in
Section~\ref{secpath}. 

Considering multicast, it is important to decide whether or not all receivers
(called sinks) can recover all the transmitted information from the senders
(called sources).  It is also important to decide the minimum size $q$ of the
finite field $\F_q$ required for linear network coding.

Before using linear network coding we have to decide coefficients in linear
combinations computed by intermediate nodes. When the size $q$ of a finite field
is large, it is shown that random choice of coefficients allows all sinks to
recover the original transmitted information with high probability \cite{ho06}.
Such a method is called random linear network coding and the probability is
called success probability. As to random linear network coding the estimation or
determination of the success probability is very important.  Ho
et~al.~\cite{ho06} gave a lower bound on the success probability.

In their paper \cite{koetter03}, Koetter and M\'edard introduced an algebraic
geometric point view on network coding.  As explained in \cite{bn:cox},
computational problems in algebraic geometry can often be solved by Gr\"obner
bases. In this paper, we shall show that the exact computation of the minimum
$q$ can be made by applying the division algorithm for multivariate polynomials,
and we will show that improved estimates for the success probability can be
found by applying the footprint bound from Gr\"obner basis theory. These results
introduce a new approach to network coding study. As the improved estimates take
into account which monomials occur in the support of the determinant of a
certain matrix \cite{ho06} we study this matrix in details at the end of the
paper.

\section{Preliminary}\label{secpre}

We can determine whether or not all sinks can recover all the transmitted
information by the determinant of some matrix \cite{ho06}.  We shall review the
definition of such determinant.  Let $G=(V,E)$ be an directed acyclic graph with
possible parallel edges that represents the network topology.  The set of source
and sink nodes is denoted by $S$ and $T$ respectively. Assume that the source
nodes $S$ together get $h$ symbols in $\F_q$ per unit time and try to send them.

Identify the edges in $E$ with the integers $1, \ldots, |E|$. For an edge
$j=(u,v)$ we write $\text{head}(j)=v$ and $\text{tail}(j)=u$. We define the $|E|
\times |E|$ matrix $F=( f_{i,j})$ where  $f_{i,j}$ is a variable if
$\text{head}(i)=\text{tail}(j)$ and $f_{i,j}=0$ otherwise. The variable
$f_{i,j}$ is the coding coefficient from $i$ to $j$. 

Index $h$ symbols in $\F_q$ sent by $S$ by $1, \ldots, h$.  We also define an $h
\times |E|$ matrix $A=(a_{i,j})$ where $a_{i,j}$ is a variable if the edge $j$
is an outgoing edge from the source $s \in S$ sending the $i$-th symbol and
$a_{i,j}=0$ otherwise. Variables $a_{i,j}$ represent how the source nodes send
information to their outgoing edges.

Let $X(l)$ denote the $l$-th symbol generated by the sources $S$, and let $Y(j)$
denote the information sent along edge $j$.  The model is described by the
following relation
\begin{equation*}
    Y(j)
    = \sum_{i=1}^h a_{i,j} X(i)+\sum_{i: \text{head}(i)=\text{tail}(j)}
        f_{i,j}Y(i)
    .
\end{equation*}

For each sink $t \in T$ define an $h \times |E|$ matrix $B_t$ whose $(i,j)$
entry $b_{t,i,j}$ is a variable if $\text{head}(j)=t$ and equals $0$ otherwise.
The index $i$ refers to the $i$-th symbol sent by one of the sources.  Thereby
variables $b_{t,i,j}$ represent how the sink $t$ process the received data from
its incoming edges.

The sink $t$ records the vector
\begin{equation*}
    \vec{b^{(t)}} = \big( b_1^{(t)}, \ldots , b_h^{(t)} \big)
\end{equation*}
where
\begin{equation*}
    b_i^{(t)} = \sum_{j: \text{head}(j)=t} b_{t,i,j}Y(j)
    .
\end{equation*}
We now recall from~\cite{ho06} under which conditions all informations sent by
the sources can always be recovered at all sinks. As in \cite{ho06} we define
the Edmonds matrix $M_t$ for $t\in T$ by
\begin{equation}
    \label{eqMt}
    M_t =
    \begin{pmatrix}
        A & 0 \\
        I-F & B_t^T
    \end{pmatrix}
    .
\end{equation}
Define the polynomial $P$ by
\begin{equation}
    \label{eqP}
    P = \prod_{t\in T} |M_t|
    .
\end{equation}
$P$ is a multivariate polynomial in variables $f_{i,j}$, $a_{i,j}$ and
$b_{t,i,j}$. Assigning a value in $\F_q$ to each variable corresponds to
choosing a coding scheme. Plugging the assigned values into $P$ gives an element
$k \in \F_q$. The following theorem from~\cite{ho06} tells us when the coding
scheme can be used to always recover the information generated at the sources
$S$ at all sinks in $T$.

\begin{theorem}
    Let the notation and the network coding model be as above. Assume a coding
    scheme has been chosen by assigning values to the variables $f_{i,j}$,
    $a_{i,j}$ and $b_{t,i,j}$. Let $k$ be the value found by plugging the
    assigned values into $P$. Every sink $t \in T$ can recover from
    $\vec{b^{(t)}}$ the informations $X(1), \ldots , X(h)$ no matter what they
    are, if and only if $k \neq 0$ holds.
\end{theorem}

\begin{proof}
    See~\cite{ho06}. 
\end{proof}

\section{Computation of the Minimum Field Size}\label{secfieldsize}

We shall study computation of the minimum symbol size $q$. For this purpose we
will need the division algorithm for multivariate polynomials
\cite[Sec.~2.3]{bn:cox} to produce the remainder of a polynomial $F(X_1, \ldots
, X_n)$ modulo $(X_1^q-X_1, \ldots, X_n^q-X_n)$ (this remainder is independent
of the choice of monomial ordering). We adapt the standard notation for the
above remainder which is
\begin{equation*}
    F(X_1, \ldots, X_n) \rem (X_1^q-X_1, \ldots, X_n^q-X_n)
    .
\end{equation*}
The reader unfamiliar with the division algorithm can think of the above
remainder of $F(X_1, \ldots, X_n)$ as the polynomial produced by the following
procedure. As long as we can find an $X_i$ such that $X_i^q$ divides some term
in the polynomial under consideration we replace the factor $X_i^q$ with $X_i$
wherever it occurs. The process continues until the $X_i$-degree is less than
$q$ for all $i=1, \ldots, n$. It is clear that the above procedure can be
efficiently implemented.

\begin{proposition}\label{pro1}
    Let $F(X_1$, \ldots, $X_n)$ be an $n$-variate polynomial over $\F_q$.  There
    exists an $n$-tuple $(x_1, \ldots, x_n) \in \F_q^n$ such that $F(x_1,
    \ldots, x_n) \neq 0$ if and only if
    \begin{equation*}
        F(X_1, \ldots, X_n) \rem (X_1^q-X_1, \ldots, X_n^q-X_n)
        \neq 0
        .
    \end{equation*}
\end{proposition}

\begin{proof}
    As $a^q=a$ for all $a \in \F_q$ it holds that  $F(X_1, \ldots, X_n)$
    evaluates to the same as $R(X_1, \ldots, X_n):=F(X_1, \ldots, X_n) \rem
    (X_1^q-X_1, \ldots, X_n^q-X_n)$ in every $(x_1, \ldots, x_n) \in \F_q^n$.
    If  $R(X_1, \ldots, X_n)=0$ therefore $F(X_1, \ldots, X_n)$ evaluates to
    zero for every choice of $(x_1, \ldots, x_n) \in \F_q^n$. If $R(X_1, \ldots,
    X_n)$ is nonzero we consider it first as a polynomial in $\F_q(X_1, \ldots,
    X_{n-1})[X_n]$ (that is, a polynomial in one variable over the quotient
    field $\F_q(X_1, \ldots, X_{n-1})$). But the $X_n$-degree is at most $q-1$
    and therefore it has at most $q-1$ zeros. We conclude that there exists an
    $x_n\in \F_q$ such that $R(X_1, \ldots, X_{n-1},x_n) \in \F_q[X_1, \ldots,
    X_{n-1}]$ is nonzero. Continuing this way we find $(x_1, \ldots, x_n)$ such
    that $R(x_1, \ldots, x_n)$ and therefore also $F(x_1, \ldots, x_n)$ is
    nonzero.
\end{proof}

From~\cite[Th.\ 2]{ho06} we know that for all prime powers $q$ greater than
$|T|$ linear network coding is possible. It is now straightforward to describe
an algorithm that finds the smallest field $\F_q$ of prescribed characteristic
$p$ for which linear network coding is feasible. We first reduce the polynomial
$P$ from~\eqref{eqP} modulo the prime $p$.  We observe that although $P$ is a
polynomial in all the variables $a_{i,j}$, $b_{t,i,j}$, $f_{i,j}$ the variable
$b_{t,i,j}$ appears at most in powers of $1$. This is so as it appears at most
in a single entry in $M_t$ and does not appear elsewhere.  Therefore $\F_q$ can
be used for network coding if $P \rem p$ does not reduce to zero modulo the
polynomials $a_{i,j}^q-a_{i,j}$, $f_{i,j}^q-f_{i,j}$. To decide the smallest
field $\F_q$ of characteristic $p$ for which network coding is feasible we try
first $\F_q=\F_p$. If this does not work we then try $\F_{p^2}$ and so on. To
find an $\F_q$ that works we need at most to try $\lfloor \log_p(|T|)\rfloor$
different fields as we know that linear network coding is possible whenever $q >
|T|$.

Note that once a field $\F_q$ is found such that the network connection problem
is feasible the last part of the proof of Proposition 1 describes a simple way
of deciding coefficients $(x_1, \ldots, x_n) \in \F_q^n$ that can be used for
network coding. 

From~\cite[Sec.~7.1.3]{frasol} we know that it is an NP-hard problem to find the
minimum field size for linear network coding. Our findings imply that it is
NP-hard to find the polynomial $P$ in~\eqref{eqP}.

\section{Computation of the Success Probability of Random Linear Network
Coding}\label{secrandom}

In random linear network coding we from the beginning fix for a collection
\begin{equation*}
    K \subseteq \{1, \ldots, h\} \times \{1, \ldots, | E | \}
\end{equation*}
the $a_{i,j}$'s with $(i,j) \in K$ and also we fix  for a collection
\begin{equation*}
    J \subseteq \{1, \ldots, | E| \} \times \{1, \ldots,|E| \}
\end{equation*}
the $f_{i,j}$'s with $(i,j) \in J$. This is done in a way such that there exists
a solution to the network connection problem with the same values for these
fixed coefficients. A priori of course we let $a_{i,j}=0$ if the edge $j$ is not
emerging from the source sending information $i$, and also a priori we of course
let $f_{i,j}=0$ if $j$ is not an adjacent downstream edge of $i$.  Besides these
a priori fixed values there may be good reasons for also fixing other
coefficients $a_{i,j}$ and $f_{i,j}$ \cite{ho06}. If for example there is only
one upstream edge $i$ adjacent to $j$ we may assume $f_{i,j}=1$.  All the
$a_{i,j}$'s and $f_{i,j}$'s which have not been fixed at this point are then
chosen randomly and independently. All coefficients are to be elements in
$\F_q$. If a solution to the network connection problem exists with the
$a_{i,j}$'s and the $f_{i,j}$'s specified,  it is possible to determine values
of $b_{t,i,j}$ at the sinks such that a solution to the network connection
problem is given. Let $\mu$ be the number of variables $a_{i,j}$ and $f_{i,j}$
chosen randomly. Call these variables $X_1, \ldots, X_\mu$. Consider the
polynomial $P$ in~\eqref{eqP} and let $\tilde{P}$ be the polynomial made from
$P$  by plugging in the fixed values of the $a_{i,j}$'s and the fixed values of
the $f_{i,j}$'s (calculations taking place in $\F_q$). Then $\tilde{P}$ is a
polynomial in $X_1, \ldots, X_\mu$. The coefficients of $\tilde{P}$ are
polynomials in the $b_{t,i,j}$'s over $\F_q$. Finally, define
\begin{equation*}
    \widehat{P} := \tilde{P} \rem (X_1^q-X_1, \ldots, X_\mu^q-X_\mu)
    .
\end{equation*}
The success probability of random linear network coding is the probability that
the random choice of coefficients will lead to a solution of the network
connection problem\footnote{This corresponds to saying that each sink can
recover the data at the maximum rate promised by network coding.} as in
Section~\ref{secpre}. That is, the probability is the number 
\begin{align}
    \nonumber
    &\frac{| \{ (x_1, \ldots, x_\mu) \in \F_q^\mu | \tilde{P}(x_1, \ldots,
        x_\mu) \neq 0 \} |}{q^\mu}
    \\
    \label{eqprob2}
    &=\frac{| \{ (x_1, \ldots, x_\mu) \in \F_q^\mu | \widehat{P}(x_1, \ldots,
        x_\mu) \neq 0 \} |}{q^\mu}
    .
\end{align}

To see the first result observe that for fixed $(x_1, \ldots, x_\mu) \in
\F_q^\mu$, $\tilde{P}(x_1, \ldots, x_\mu)$ can be viewed as a polynomial in the
variables $b_{t,i,j}$'s with coefficients in $\F_q$ and recall that the
$b_{t,i,j}$'s occur in powers of at most $1$. Therefore, if $\tilde{P}(x_1,
\ldots, x_\mu) \neq 0$, then by Proposition~\ref{pro1} it is possible to choose
the $b_{t,i,j}$'s such that if we plug them into $\tilde{P}(x_1, \ldots, x_\mu)$
then we get nonzero.  The last result follows from the fact that $\tilde{P}(x_1,
\ldots, x_\mu) = \widehat{P}(x_1, \ldots, x_\mu)$ for all $(x_1, \ldots, x_\mu)
\in \F_q^\mu$. In this section we shall present a method to estimate the success
probability using Gr\"obner basis theoretical methods. 

We briefly review some basic definitions and results of Gr\"obner bases.  See
\cite{bn:cox} for a more detailed exposition.  Let $\mathcal{M}(X_1, \ldots,
X_n)$ be the set of monomials in the variables $X_1, \ldots, X_n$.  A monomial
ordering $\prec$ is a total ordering on $\mathcal{M}(X_1, \ldots, X_n)$ such
that
\begin{equation*}
    L \prec M \Longrightarrow LN \prec MN
\end{equation*}
holds for all monomials $L$, $M$, $N \in \mathcal{M}(X_1, \ldots, X_n)$ and such
that every nonempty subset of $\mathcal{M}(X_1, \ldots, X_n)$ has a unique
smallest element with respect to $\prec$.  The leading monomial of a polynomial
$F$ with respect to $\prec$, denoted by $\textsc{lm}(F)$, is the largest
monomial in the support of $F$. Given a polynomial ideal $I$ and a monomial
ordering the footprint $\Delta_{\prec}(I)$ is the set of monomials that cannot
be found as leading monomials of any polynomial in $I$. The following
proposition explains our interest in the footprint (for a proof of the
proposition see~\cite[Pro.~8.32]{beckerweis}).

\begin{proposition}\label{profootgen}
    Let $\F$ be a field and consider the polynomials $F_1, \ldots, F_s \in
    \F[X_1, \ldots, X_n]$. Let $I=\langle F_1, \ldots, F_s \rangle \subseteq
    \F[X_1, \ldots, X_n]$ be the ideal generated by $F_1, \ldots, F_s$. If
    $\Delta_{\prec}(I)$ is finite then the number of common zeros of $F_1,
    \ldots, F_s$ in the algebraic closure of $\F$ is at most equal to $|
    \Delta_{\prec}(I)|$. 
\end{proposition}

Proposition~\ref{profootgen} is known as the footprint bound. It
has the following corollary.

\begin{corollary}\label{corweakfoot}
    Let $F \in \F[X_1, \ldots,X_n]$ where $\F$ is a field containing $\F_q$.
    Fix a monomial ordering and let 
    \begin{equation*}
        X_1^{j_1}\cdots X_n^{j_n}
        = \textsc{lm}\big( F \rem (X_1^q-X_1, \ldots, X_n^q-X_n) \big)
        .
    \end{equation*}
    The number of zeros of $F$ over $\F_q$ is at most equal to
    \begin{equation}
        \label{eqcirc1}
        q^n-\prod_{v=1}^n(q-j_v)
        .
    \end{equation}
\end{corollary}

\begin{proof}
    We have
    \begin{multline*}
        \Delta_{\prec}(\langle F, X_1^q-X_1, \ldots, X_n^q-X_n \rangle )
        \\
        \subseteq \Delta_{\prec}(\langle \textsc{lm}(F \rem (X_1^q-X_1,
        \ldots,X_n^q-X_n)),X_1^q, \ldots, X_n^q \rangle )
    \end{multline*}
    and the size of the latter set equals~\eqref{eqcirc1}. The result now
    follows immediately from Proposition~\ref{profootgen}.
\end{proof}

\begin{theorem}\label{theestimate}
    Let as above $\tilde{P}$ be found by plugging into $P$ some fixed values for
    the variables $a_{i,j}$, $(i,j) \in K$, and by plugging into $P$ some fixed
    values for the variables $f_{i,j}$, $(i,j) \in J$, and by leaving the
    remaining $\mu$ variables flexible. Assume as above that there exists a
    solution to the network connection problem with the same values for these
    fixed coefficients. Denote by $X_1, \ldots, X_\mu$ the variables to be
    chosen by random and define $\widehat{P}:= \tilde{P} \rem (X_1^q-X_1,
    \ldots, X_\mu^q-X_\mu)$. (Note that if $q > |T|$ then
    $\widehat{P}=\tilde{P}$). Consider $\widehat{P}$ as a polynomial in the
    variables $X_1, \ldots, X_\mu$  and let $\prec$ be any fixed monomial
    ordering. Writing $X_1^{j_1}\cdots X_\mu^{j_\mu} = \textsc{lm}(\widehat{P})$
    the success probability is at least
    \begin{equation}
        \label{eqtrio}
        q^{- \mu} \prod_{v=1}^\mu (q-j_v)
        .
    \end{equation}
    As a consequence the success probability is in particular at least   
    \begin{equation}
        \label{eqtri}
        q^{- \mu} \min \bigg\{
        \prod_{i=1}^\mu (q-s_i)
        \bigg|
        X_1^{s_1}\cdots X_\mu^{s_\mu} \text{ is a monomial in the
        support of } \widehat{P}
        \bigg\}
        .
    \end{equation}
\end{theorem}

\begin{proof}
    Let $\F$ be the quotient field $\F_q(X_1, \ldots, X_\mu)$. The result
    in~\eqref{eqtrio} now follows by applying Corollary~\ref{corweakfoot}
    and~\eqref{eqprob2}. As the leading monomial of $\tilde{P}$ is of course a
    monomial in the support of $\tilde{P}$ \eqref{eqtri} is smaller or equal
    to~\eqref{eqtrio}.
\end{proof}

\begin{remark}
    The condition in Theorem~\ref{theestimate} that there exists a solution to
    the network connection problem with the coefficients corresponding to $K$
    and $J$ being as specified is equivalent to the condition that $\widehat{P}
    \neq 0$.
\end{remark}

We conclude this section by mentioning without a proof that Gr\"obner basis
theory tells us that the true success probability can be calculated as 
\begin{equation*}
    q^{-\mu} \big( q^\mu-| \Delta_{\prec}(\langle \tilde{P}, X_1^q-X_1, \ldots,
    X_\mu^q-X_\mu \rangle )| \big)
    .
\end{equation*}
This observation is however of little value as it seems very difficult to
compute the footprint  
\begin{equation*}
    \Delta_{\prec}(\langle \tilde{P}, X_1^q-X_1, \ldots, X_\mu^q-X_\mu \rangle)
\end{equation*}
due to the fact that $\mu$ is typically a very high number.

\section{The Bound by Ho et al.}\label{secho}

In~\cite{ho06} Ho et al.\ gave a lower bound on the success probability in terms
of the number of edges $j$ with associated random coefficients\footnote{We state
Ho et al.'s bound only in the case of delay-free acyclic networks.} $\{a_{i,j},
f_{l,j}\}$. Letting $\eta$ be the number of such edges \cite[Th.\ 2]{ho06} tells
us that if $q>|T|$ and if there exists a solution to the network connection
problem with the same values for the fixed coefficients, then the success
probability is at least
\begin{equation}
    \label{eqhosnabel1}
    p_\text{Ho} = \left( \frac{q-|T|}{q} \right)^\eta
    .
\end{equation}
The proof in~\cite{ho06} of~\eqref{eqhosnabel1} relies on two lemmas of which we
only state the first one.

\begin{lemma}\label{lemho3}
    Let $\eta$ be defined as above. The determinant polynomial of $M_t$ has
    maximum degree $\eta$ in the random variables $\{a_{i,j}, f_{l,j} \}$ and is
    linear in each of these variables.
\end{lemma}

\begin{proof}
    See~\cite[Lem.\ 3]{ho06}. Alternatively the proof can be derived as a
    consequence of Theorem~\ref{thepath} in Section~\ref{secpath}.
\end{proof}

Recall, that the polynomial $P$ in~\eqref{eqP} is the product of the
determinants $|M_t|$, $t\in T$. Lemma~\ref{lemho3} therefore implies that the
polynomial $\tilde{P}$ has at most total degree equal to $|T| \eta$ and that no
variable appears in powers of more than $|T|$. The assumption $q > |T|$ implies
$\widehat{P}=\tilde{P}$ which makes it particular easy to see that the same of
course holds for $\widehat{P}$.  Combining this observation with the following
lemma shows that the numbers in~\eqref{eqtrio} and \eqref{eqtri} are both at
least as large as the number~\eqref{eqhosnabel1}.

\begin{lemma}
    Let $\eta, |T|, q \in \mathbf{N}$, $|T| < q$ be some fixed numbers. Let
    $\mu, x_1, \ldots, x_\mu \in \mathbf{N}_0$ satisfy
    \begin{equation*}
        0 \leq x_1 \leq |T|, \ldots, 0 \leq x_\mu \leq |T|
    \end{equation*}
    and $x_1 + \cdots + x_\mu \leq |T|\eta$. The minimal value of
    \begin{equation*}
        \prod_{i=1}^\mu \left( \frac{q-x_i}{q} \right)
    \end{equation*}
    (taken over all possible values of $\mu, x_1, \ldots, x_\mu$) is
    \begin{equation*}
        \left( \frac{q-|T|}{q} \right)^\eta
        .
    \end{equation*}
\end{lemma}

\begin{proof}
    Assume $\mu$ and $x_1, \ldots, x_\mu$ are chosen such that the expression
    attains its minimal value. Without loss of generality we may assume that
    \begin{equation*}
        x_1 \geq x_2 \geq \cdots \geq x_\mu
    \end{equation*}
    holds. Clearly, $x_1+\cdots + x_\mu=|T|\eta$ must hold. If $x_i<|T|$ and
    $x_{i+1}>0$ then 
    \begin{equation*}
        (q-x_i)(q-x_{i+1} )> (q-(x_i+1))(q-(x_{i+1}-1))
    \end{equation*}
    which cannot be the case. So $x_1=\cdots =x_{\eta}=|T|$. The remaining
    $x_j$'s if any all equal zero.
\end{proof}

\section{Examples}

In this section we apply the methods from the previous sections to two concrete
networks. We will see that the estimate on the success probability of random
linear network coding that was described in Theorem~\ref{theestimate} can be
considerably better than the estimate described in~\cite[Th.\ 2]{ho06}. Also we
will apply the method from Section~\ref{secfieldsize} to determine the smallest
field of characteristic two for which network coding can be successful.

As random linear network coding is assumed to take place at the nodes in a
decentralized manner, one natural choice is to set $f_{i,j}=1$ whenever the
indegree of the end node of edge $i$ is one and $j$ is the downstream edge
adjacent to $i$. Clearly, if $j$ is not a downstream edge adjacent to $i$ we set
$f_{i,j}=0$. Whenever none of the above is the case we may choose $f_{i,j}$
randomly. Also if there is only one source and the outdegree of the source is
equal to the number of symbols to be send we may enumerate the edges from the
source by the numbers $1, \ldots, h$ and set $a_{i,j}=1$ if $1 \leq i=j \leq h$
and set $a_{i,j}=0$ otherwise. This strategy can be generalized also to deal
with the case of more sources. In the following two examples we will choose the
variables in the manner just described. The network in the first example is
taken from~\cite[Ex.~3.1]{frasol} whereas the network in the second example is
new.

\begin{example}\label{exsidstenye}
    Consider the delay-free and acyclic network in Figure~\ref{figny}. 
    \begin{figure}
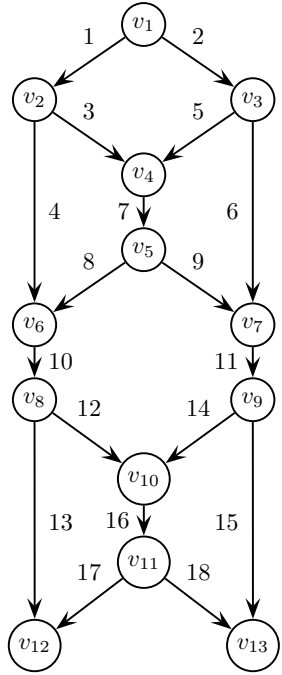

        \begin{center}
            \scalebox{0.909}{
            $
            \psmatrix[colsep=0.8cm,rowsep=.4cm,mnode=circle]
            &v_1\\
            v_2&&v_3\\
            &v_4\\
            &v_5\\
            v_6&&v_7\\
            v_8&&v_9\\
            &v_{10}\\
            &v_{11}\\
            v_{12}&&v_{13}
            \psset{arrowscale=2}
            \ncline{->}{1,2}{2,1}^{1}
            \ncline{->}{1,2}{2,3}^{2}
            \ncline{->}{2,1}{3,2}^{3}
            \ncline{->}{2,1}{5,1}>{4}
            \ncline{->}{2,3}{3,2}^{5}
            \ncline{->}{2,3}{5,3}<{6}
            \ncline{->}{3,2}{4,2}<{7}
            \ncline{->}{4,2}{5,1}^{8}
            \ncline{->}{4,2}{5,3}^{9}
            \ncline{->}{5,1}{6,1}>{10}
            \ncline{->}{5,3}{6,3}<{11}
            \ncline{->}{6,1}{7,2}^{12}
            \ncline{->}{6,1}{9,1}>{13}
            \ncline{->}{6,3}{7,2}^{14}
            \ncline{->}{6,3}{9,3}<{15}
            \ncline{->}{7,2}{8,2}<{16}
            \ncline{->}{8,2}{9,1}^{17}
            \ncline{->}{8,2}{9,3}^{18}
            \endpsmatrix
            $
            }
        \end{center}
        \caption{The network from Example~\ref{exsidstenye}}
        \label{figny}
    \end{figure}
    There is one sender $v_1$ and two receivers $v_{12}$ and $v_{13}$. The
    min-cut max-flow number is two for both receivers so we assume that two
    independent random processes emerge from sender $v_1$.  We consider in this
    example only fields of characteristic $2$. Following the description
    preceding the example we set $a_{1,1}=a_{2,2}=1$ and $a_{i,j}=0$ in all
    other cases. Also we let $f_{i,j}=1$ except 
    \begin{equation*}
        f_{3,7}, f_{5,7}, f_{4,10}, f_{8,10}, f_{9,11}, f_{6,11},
        f_{12,16}, f_{14,16}
    \end{equation*}
    which we choose by random. As in the previous sections we consider
    $b_{t,i,j}$ as fixed but unknown to us. The determinant polynomial becomes
    \begin{equation*}
        \tilde{P}=(b^2c^2e^2gh+c^2f^2gh+a^2d^2f^2gh)Q
        ,
    \end{equation*}
    where 
    \begin{align*}
        a &=f_{3,7}  & b &=f_{5,7}  & c &=f_{4,10}  & d &=f_{8,10} \\
        e &=f_{9,11} & f &=f_{6,11} & g &=f_{12,16} & h &=f_{14,16}
    \end{align*}
    and $Q=|B_{v_{12}}'| \, |B_{v_{13}}'|$. Here, $B_{v_{12}}'$ respectively
    $B_{v_{13}}'$ is the matrix consisting of the nonzero columns of
    $B_{v_{12}}$ respectively the nonzero columns of $B_{v_{14}}$.  Restricting
    to fields $\F_q$ of size at least $4$ we have $\widehat{P}=\tilde{P}$ and we
    can therefore immediately apply the bounds in Theorem~\ref{theestimate}.
    Applying~\eqref{eqtri} we get the following lower bound on the success
    probability
    \begin{equation*}
        P_\text{new 2}(q) = \frac{(q-2)^3(q-1)^2}{q^5}
        .
    \end{equation*}
    Choosing as monomial ordering the lexicographic ordering $\plex$ with 
    \begin{equation*}
        a \plex b \plex d \plex e \plex g \plex h \plex f \plex c
    \end{equation*}
    the leading monomial of $\tilde{P}$ becomes $c^2f^2gh$ and therefore
    from~\eqref{eqtrio} we get the following lower bound on the success
    probability
    \begin{equation*}
        P_\text{new 1}(q) = \frac{(q-2)^2(q-1)^2}{q^4}
        .
    \end{equation*}
    For comparison the bound~\eqref{eqhosnabel1} from~\cite{ho06} states that
    the success probability is at least
    \begin{equation*}
        P_{\text{Ho}}(q)=\frac{(q-2)^4}{q^4}
        .
    \end{equation*}
    We see that $P_\text{new 1}$ exceeds $P_\text{Ho}$ with a factor
    $(q-1)^2/(q-2)^2$, which is larger than 1. Also $P_\text{new 2}$ exceeds
    $P_\text{Ho}$. In Table~\ref{tabeins} we list values of $P_\text{new 1}(q)$,
    $P_\text{new 2}(q)$ and $P_\text{Ho}(q)$ for various choices of $q$. 

    \begin{table}
        \centering
        \caption{From Example~\ref{exsidstenye}: Estimates on the success
        probability} 
        \label{tabeins}
        \begin{tabular}{r ccccc}
            \toprule
            $q$ & 4 & 8 & 16 & 32 & 64
            \\
            $P_\text{new 1}(q)$ & 0.140 & 0.430 & 0.672 & 0.893 & 0.909
            \\
            $P_\text{new 2}(q)$ & $0.703 \times 10^{-1}$ & 0.322 & 0.588 &
            0.773 & 0.880
            \\
            $P_\text{Ho}(q)$  & $0.625 \times 10^{-1}$ & 0.316 & 0.586 & 0.772
            & 0.880
            \\
            \bottomrule
        \end{tabular}
    \end{table}
    We next consider the field $\F_2$. We reduce $\tilde{P}$ modulo $(a^2-a,
    \ldots, h^2-h)$ to get
    \begin{equation*}
        \widehat{P} = (bcegh + cfgh + adfgh) Q
        .
    \end{equation*}
    From~\eqref{eqtri} we see that the success probability of random network
    coding is at least $2^{-5}$. Choosing as monomial ordering the lexicographic
    ordering described above~\eqref{eqtrio} tells us that the success
    probability is at least $2^{-4}$. For comparison the
    bound~\eqref{eqhosnabel1} does not apply as we do not have $q>|T|$. It
    should be mentioned that for delay-free acyclic networks the network coding
    problem is solvable for all choices of $q\geq |T|$ \cite{jaggi03} and
    \cite{sanders03}. From this fact one can only conclude that the success
    probability is at least $2^{-8}$ (8 being the number of coefficients to be
    chosen by random).
\end{example}

\begin{example}\label{exone}
    Consider the network in Figure~\ref{figone}.
    \begin{figure}
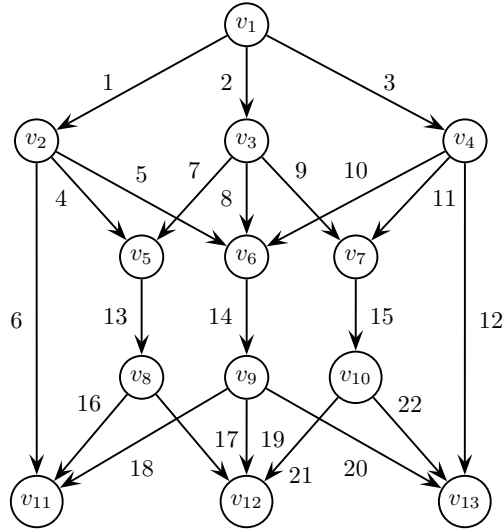

        \begin{center}
            \scalebox{0.909}{
            $
            \psmatrix[colsep=0.8cm,rowsep=1cm,mnode=circle]
            &&v_1\\
            v_2&&v_3&&v_4\\
            &v_5&v_6&v_7\\
            &v_8&v_9&v_{10}\\
            v_{11}&&v_{12}&&v_{13}
            \psset{arrowscale=2} 
            \ncline{->}{1,3}{2,1}<{1}
            \ncline{->}{1,3}{2,3}<{2}
            \ncline{->}{1,3}{2,5}>{3}
            \ncline{->}{2,1}{3,2}<{4}
            \ncline{->}{2,1}{3,3}^{5}
            \ncline{->}{2,1}{5,1}<{6}
            \ncline{->}{2,3}{3,2}^{7}
            \ncline{->}{2,3}{3,3}<{8}
            \ncline{->}{2,3}{3,4}^{9}
            \ncline{->}{2,5}{3,3}^{10}
            \ncline{->}{2,5}{3,4}>{11}
            \ncline{->}{2,5}{5,5}>{12}
            \ncline{->}{3,2}{4,2}<{13}
            \ncline{->}{3,3}{4,3}<{14}
            \ncline{->}{3,4}{4,4}>{15}
            \ncline{->}{4,2}{5,1}^{16}
            \ncline{->}{4,2}{5,3}>{17}
            \ncline{->}{4,3}{5,1}_{18}
            \ncline{->}{4,3}{5,3}>{19}
            \ncline{->}{4,3}{5,5}_{20}
            \ncline{->}{4,4}{5,3}_{21}
            \ncline{->}{4,4}{5,5}^{22}
            \endpsmatrix 
            $
            }
        \end{center}
        \caption{The network from Example~\ref{exone}}
        \label{figone}
    \end{figure}
    The sender $v_1$ generates 3 independent random processes. The vertices
    $v_{11}$, $v_{12}$ and $v_{13}$ are the receivers. We will apply network
    coding over various fields of characteristic two. We start by considering
    random linear network coding over fields of size at least $4$. As $4>|T|=3$
    we know that this can be done successfully.

    We set $a_{1,1}=a_{2,2}=a_{3,3}=1$ and $a_{i,j}=0$ in all other cases.  We
    let $f_{i,j}=1$ except $f_{4,13}, f_{7,13}, f_{5,14}, f_{8,14}, f_{10,14},
    f_{9,15}, f_{11,15}$, which we choose by random. As in the last section we
    consider $b_{t,i,j}$ as fixed but unknown to us. Therefore
    $\tilde{P}=\widehat{P}$ is a polynomial in the seven variables  $f_{4,13},
    f_{7,13}, f_{5,14},f_{8,14},f_{10,14},f_{9,15},f_{11,15}$.  The determinant
    polynomial becomes
    \begin{equation*}
        \widehat{P}=(abcdefg+abce^2f^2+b^2c^2efg)Q
        ,
    \end{equation*}
    where 
    \begin{align*}
        a &= f_{4,13} & b &= f_{5,14}  & c &= f_{7,13}  & d &= f_{8,14}\\
        e &= f_{9,15} & f &= f_{10,14} & g &= f_{11,15}
    \end{align*}
    and $Q=|B_{v_{11}}'| \, |B_{v_{12}}'| \, |B_{v_{13}}'|$. Here, $B_{v_{11}}'$
    respectively $B_{v_{12}}'$ respectively $B_{v_{13}}'$ is the matrix
    consisting of the nonzero columns of $B_{v_{11}}$ respectively the nonzero
    columns of $B_{v_{12}}$ respectively the nonzero columns of $B_{v_{13}}$.
    Choosing a lexicographic ordering with $d$ being larger than the other
    variables and applying~\eqref{eqtrio} we get that the success probability is
    at least 
    \begin{equation*}
        P_\text{new 1}(q)=\frac{(q-1)^7}{q^7}
        .
    \end{equation*}
    Applying~\eqref{eqtri} we see that the success probability is at least 
    \begin{equation*}
        P_\text{new 2}(q)=\frac{(q-1)^3(q-2)^2}{q^5}
        .
    \end{equation*}
    For comparison~\eqref{eqhosnabel1} tells us that success probability is at
    least
    \begin{equation*}
        P_{\text{Ho}}(q)=\frac{(q-3)^3}{q^3}
        .
    \end{equation*}
    Both bound~\eqref{eqtrio} and bound~\eqref{eqtri}
    exceed~\eqref{eqhosnabel1} for all values of $q \geq 4$.  In
    Table~\ref{tab1} we list $P_\text{new 1}(q)$, $P_\text{new 2}(q)$ and
    $P_\text{Ho}(q)$ for various values of $q$.

    \begin{table}
        \centering
        \caption{From Example~\ref{exone}: Estimates on the success
        probability}
        \label{tab1}
        \begin{tabular}{r ccccc}
            \toprule
            $q$ & 4 & 8 & 16 & 32 & 64
            \\
            $P_\text{new 1}(q)$ & 0.133 & 0.392 & 0.636 & 0.800 & 0.895
            \\
            $P_\text{new 2}(q)$ & 0.105 & 0.376 & 0.630 & 0.799 & 0.895
            \\
            $P_\text{Ho}(q)$  & $0.156 \times 10^{-1}$ & 0.244 & 0.536 & 0.744
            & 0.865
            \\
            \bottomrule
        \end{tabular}
    \end{table}

    We next consider the field $\F_2$. We reduce $\tilde{P}$ modulo $(a^2-a,
    \ldots, g^2-g)$ to get
    \begin{equation*}
        \widehat{P} = (abcdefg + abcef + bcefg) Q
        .
    \end{equation*}
    From~\eqref{eqtri} we see that the success probability of random network
    coding is at least $2^{-7}$. Choosing a proper monomial ordering we get
    from~\eqref{eqtrio}  that the success probability is at least $2^{-5}$. For
    comparison neither \cite{ho06}, \cite{jaggi03}, nor \cite{sanders03} tells
    us that linear network coding is possible.
\end{example}

\section{The Topological Meaning of $|M_t|$}\label{secpath}

Recall from Section~\ref{secho} that Ho et al.'s bound~\eqref{eqhosnabel1}
relies on the rather rough Lemma~\ref{lemho3}. The following theorem gives a
much more precise description of which monomials can occur in the support of $P$
and $\tilde{P}$ by explaining exactly which monomials can occur in $|M_t|$.
Thereby the theorem gives some insight into when the bounds \eqref{eqtrio} and
\eqref{eqtri} are much better than the bound~\eqref{eqhosnabel1}. The theorem
states that if $K$ is a monomial in the support of $|M_t|$ then it is the
product of $a_{i,j}$'s, $f_{i,j}$'s and $b_{t,i,j}$'s related to $h$ edge
disjoint paths $P_1, \ldots, P_h$ that originate in the senders and end in
receiver $t$.

\begin{theorem}\label{thepath}
    Consider a delay-free acyclic network.  If $K$ is a monomial in the support
    of the determinant of $M_t$ then it is of the form $K_1\cdots K_h$ where 
    \begin{equation*}
        K_u = a_{u,l^{(u)}_1}f_{l_1^{(u)},l_2^{(u)}}  f_{l_2^{(u)},l_3^{(u)}}
        \cdots f_{l^{(u)}_{s_u-1},l^{(u)}_{s_u}}b_{t,v_u,l^{(u)}_{s_u}}
    \end{equation*}
    for $u=1, \ldots, h$. Here, $\{v_1, \ldots, v_h\} = \{1, \ldots, h\}$ holds
    and $l_1^{(1)}, \ldots,l_h^{(h)}$ respectively $l_{s_1}^{(1)}, \ldots,
    l_{s_h}^{(h)}$ are pairwise different. Further
    \begin{equation*}
        f_{l^{(u_1)}_i,l^{(u_1)}_{i+1}} \neq f_{l^{(u_2)}_j,l^{(u_2)}_{j+1}}
    \end{equation*}
    unless $u_1=u_2$ and $i=j$ hold. In other words $K$ corresponds to a product
    of $h$ edge disjoint paths.
\end{theorem}

\begin{proof}
    A proof can be found in the appendix. 
\end{proof}

We illustrate the theorem with an example.

\begin{example}\label{exbutter}
    Consider the butterfly network in Figure~\ref{figurbutter}. 
    \begin{figure}
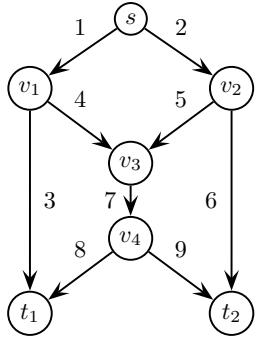

        \begin{center}
            \scalebox{0.909}{
            $
            \psmatrix[colsep=0.8cm,rowsep=0.4cm,mnode=circle]
            &s\\
            v_1&&v_2\\
            &v_3\\
            &v_4\\
            t_1&&t_2
            \psset{arrowscale=2}
            \ncline{->}{1,2}{2,1}^{1}
            \ncline{->}{2,1}{3,2}^{4}
            \ncline{->}{1,2}{2,3}^{2}
            \ncline{->}{2,3}{3,2}^{5}
            \ncline{->}{2,1}{5,1}>{3}
            \ncline{->}{2,3}{5,3}<{6}
            \ncline{->}{3,2}{4,2}<{7}
            \ncline{->}{4,2}{5,1}^{8}
            \ncline{->}{4,2}{5,3}^{9}
            \endpsmatrix 
            $
            }
        \end{center}
        \caption{The butterfly network}
        \label{figurbutter}
    \end{figure}
    A monomial $K$ is in the support of $|M_{t_1}|$ if and only if it is in the
    support of the determinant of 
    \begin{multline*}
        N_{t_1}
        = (n_{i,j})
        = \begin{bmatrix}
            I+F&B_{t_1}^T\\
            A &0
        \end{bmatrix}
        \\
        = \begin{bmatrix}
            1&0&f_{1,3}&f_{1,4}&0&0&0&0&0&0&0\\
            0&1&0&0&f_{2,5}&f_{2,6}&0&0&0&0&0\\
            0&0&1&0&0&0&0&0&0&b_{t_1,1,3}&b_{t_1,2,3}\\
            0&0&0&1&0&0&f_{4,7}&0&0&0&0\\
            0&0&0&0&1&0&f_{5,7}&0&0&0&0\\
            0&0&0&0&0&1&0&0&0&0&0\\
            0&0&0&0&0&0&1&f_{7,8}&f_{7,9}&0&0\\
            0&0&0&0&0&0&0&1&0&b_{t_1,1,8}&b_{t_1,2,8}\\
            0&0&0&0&0&0&0&0&1&0&0\\
            a_{1,1}&a_{1,2}&0&0&0&0&0&0&0&0&0\\
            a_{2,1}&a_{2,2}&0&0&0&0&0&0&0&0&0
        \end{bmatrix}
    \end{multline*}
    By inspection we see that the monomial
    \begin{equation*}
        K = a_{1,1} a_{2,2} b_{t_1,1,3} b_{t_1,2,8} f_{7,8} f_{5,7} f_{2,5}
        f_{1,3}
    \end{equation*}
    is in the support of $|N_{t_1}|$.  We can write $K = K_1 K_2$ where
    \begin{equation*}
        K_1=a_{1,1}f_{1,3}b_{t_1,1,3}
        \quad \text{and} \quad
        K_2=a_{2,2}f_{2,5}f_{5,7}f_{7,8}b_{t_1,2,8}
        .
    \end{equation*}
    This is the description guaranteed by Theorem~\ref{thepath}. To make it
    easier for the reader to follow the proof of Theorem~\ref{thepath} in the
    appendix we now introduce some of the notations to be used there. By
    inspection the monomial $K$ can be written
    \begin{equation*}
        K = \prod_{i=1}^{11} n_{i,p(i)}
    \end{equation*}
    where the permutation $p$ is given by
    \begin{align*}
        p(1) &= 3 & p(2) &= 5  & p(3) &= 10 & p(4)  &= 4 & p(5)  &= 7 & p(6)=6\\
        p(7) &= 8 & p(8) &= 11 & p(9) &= 9  & p(10) &= 1 & p(11) &= 2
    \end{align*}
    Therefore if we index the elements in $\{1, \ldots, 11\}$ by 
    \begin{align*}
        i_1 &= 10 & i_2 &=1 & i_3 &= 3 & i_4    &= 11 & i_5    &= 2 & i_6 &= 5\\
        i_7 &= 7  & i_8 &=8 & i_9 &= 4 & i_{10} &= 6  & i_{11} &= 9
    \end{align*}
    then we can write
    \begin{align*}
        K_1 &= n_{i_1,p(i_1)} n_{i_2,p(i_2)} n_{i_3,p(i_3)}
        \\
        K_2 &= n_{i_4,p(i_4)} n_{i_5,p(i_5)} n_{i_6,p(i_6)} n_{i_7,p(i_7)}
        n_{i_8,p(i_8)}
    \end{align*}
    and we have
    \begin{equation*}
        n_{i_9,p(i_9)}=n_{i_{10},p(i_{10})}=n_{i_{11},p(i_{11})} = 1
    \end{equation*}
    corresponding to the fact $p(i_9)=i_9$, $p(i_{10})=i_{10}$ and
    $p(i_{11})=i_{11}$. 
\end{example}

\begin{remark}
    The procedures described in the proof of Theorem~\ref{thepath} can be
    reversed. This implies that there is a bijective map between the set of edge
    disjoint paths $P_1, \ldots, P_h$ in Theorem~\ref{thepath} and the set of
    monomials in $|M_t|$.
\end{remark}

Theorem~\ref{thepath} immediately applies to the situation of random network
coding if we plug into the $a_{i,j}$'s and into the $f_{t,i,j}$'s on the paths
$P_1, \ldots, P_h$ the fixed values wherever such are given. Let as in
Lemma~\ref{lemho3} $\eta$ be the number of edges for which some coefficients
$a_{i,j}, f_{i,j}$ are to be chosen by random. Considering the determinant as a
polynomial in the variables to be chosen by random with coefficients in the
field of rational expressions in the $b_{t,i,j}$'s we see that no monomial can
contain more than $\eta$ variables and that no variable occurs more than once.
This is because the paths $P_1, \ldots, P_h$ are edge disjoint. Hence,
Lemma~\ref{lemho3} is a consequence of Theorem~\ref{thepath}. 

\section{Acknowledgments}

The authors would like to thank the anonymous referees for their helpful
suggestions.
 
\appendix

\section{Proof of Theorem~\ref{thepath}}

The proof of Theorem~\ref{thepath} calls for the following technical lemma.
\begin{lemma}\label{lemnocycle}
    Consider a delay-free acyclic network with corresponding matrix $F$ as in
    Section~\ref{secpre}. Let $I$ be the $|E| \times |E|$ identity matrix and
    define
    \begin{equation*}
        \Gamma=(\gamma_{i,j})=I+F
        .
    \end{equation*}
    Given a permutation $p$ on $\{1, \ldots, |E|\}$ write 
    \begin{equation*}
        p^{(i)}(\lambda) = \overbrace{p(p(\cdots(\lambda )\cdots ))}^\text{i
        times}
    \end{equation*}
    If for some $\lambda \in \{1, \ldots, |E|\}$ the following hold
    \begin{enumerate}[label=(\arabic*)]
        \item \label{enu:lemma:a}
            $\lambda,p(\lambda), \ldots,p^{(x)}(\lambda)$ are pairwise different
            
        \item \label{enu:lemma:b}
            $p^{(x+1)}(\lambda) \in \{ \lambda, p(\lambda), \ldots,
            p^{(x)}(\lambda) \}$
        \item \label{enu:lemma:c}
            $\gamma_{\lambda,p(\lambda)}, \gamma_{p(\lambda),p(p(\lambda))},
            \ldots, \gamma_{p^{(x)}(\lambda),p^{(x+1)}(\lambda)}$ are all
            nonzero
    \end{enumerate}
    then $x=0$.
\end{lemma}

\begin{proof}
    Let $p$ be a permutation and let $x$ and $\lambda$ be numbers such that
    \ref{enu:lemma:a}, \ref{enu:lemma:b} and \ref{enu:lemma:c} hold. As $p$ is a
    permutation then \ref{enu:lemma:a} and \ref{enu:lemma:b} implies that
    $p(p^{(x)}(\lambda))=\lambda$. Aiming for a contradiction assume $x>0$. As
    $p(\eta)=\eta$ does not hold for any $\eta \in \{\lambda, p(\lambda),
    \ldots, p^{(x)}(\lambda)\}$,
    \begin{equation*}
        \gamma_{\lambda,p(\lambda)}, \gamma_{p(\lambda),p^{(2)}(\lambda)},
        \ldots, \gamma_{p^{(x)}(\lambda),p^{(x+1)}(\lambda)}
    \end{equation*}
    are all non-diagonal elements in $I+F$. By \ref{enu:lemma:c} we therefore
    have constructed a cycle in a cycle-free graph and the assumption $x>0$
    cannot be true.
\end{proof}

\begin{proof}[of Theorem~\ref{thepath}]
    A monomial is in the support of the determinant of $M_t$ if and only if it
    is in the support of the determinant of
    \begin{equation*}
        N_t =
        \begin{pmatrix}
            I+F&B_t^T\\
            A&0
        \end{pmatrix}
        = (n_{i,j})
        .
    \end{equation*}
    To ease the notation in the present proof we consider the latter matrix.
    Let $p$ be a permutation on $\{1, \ldots, |E|+h\}$ such that 
    \begin{equation}
        \label{diamond1}
        \prod_{s=1}^{|E|+h}n_{s,p(s)} \neq 0
        .
    \end{equation} 
    Below we order the elements in $\{1, \ldots, |E|+h\}$ in a particular way by
    indexing them $i_1, \ldots, i_{|E|+h}$ according to the following set of
    procedures.

    Let $i_1=|E|+1$  and define recursively 
    \begin{equation*}
        i_s=p(i_{s-1})
    \end{equation*}
    until $|E| <p(i_s) \leq |E|+h$. Note that this must eventually happen due to
    Lemma~\ref{lemnocycle}.  Let $s_1$ be the (smallest) number such that $|E| <
    p(i_{s_1}) \leq |E|+h$ holds. This corresponds to saying that $n_{i_1,
    p(i_1)}$ is an entry in $A$, that $n_{i_2, p(i_2)},$ $\ldots, n_{i_{s_1-1},
    p(i_{s_1-1})}$ are entries in $I+F$, and that $n_{i_{s_1},p(i_{s_1})}$ is an
    entry in $B_t^T$.  Observe, that $p(i_r)=i_r$ cannot happen for $2 \leq r
    \leq s_1$ as already $p(i_{r-1})=i_r$ holds. As $n_{i_r,p(i_r)}$ is non-zero
    by~\eqref{diamond1} we therefore must have 
    \begin{equation*}
        n_{i_r,p(i_r)}=f_{i_{r},p(i_r)}=f_{i_{r},i_{r+1}}
    \end{equation*}
    for $2\leq r < s_1$.  Hence,
    \begin{equation*}
        (n_{i_1,p(i_1)}, \ldots, n_{i_{s_1},p(i_{s_1})})
        = (a_{1,i_2},f_{i_2,i_3}, \ldots, f_{i_{s_1-1}, i_{s_1}},
        b_{t,v_1,i_{s_1}})
    \end{equation*}
    for some $v_1$. Denote this sequence by $P_1$. Clearly, $P_1$ corresponds to
    the polynomial $K_1$ in the theorem.

    We next apply the same procedure as above starting with $i_{s_1+1}=|E|+2$ to
    get a sequence $P_2$ of length $s_2$. Then we do the same with
    $i_{s_1+s_2+1}=|E|+3, \ldots, i_{s_1+\cdots s_{h-1}+1}=|E|+h$ to get the
    sequences $P_3, \ldots, P_h$. For $u=2, \ldots, h$ we have
    \begin{multline*}
        P_u
        = \left(
            n_{i_{s_1+\cdots +s_{u-1}+1}, p(i_{s_1+\cdots +s_{u-1}+1})}, \ldots
            , n_{i_{s_1+\cdots +s_{u}}, p(i_{s_1+\cdots +s_{u}})}
        \right)
        \\
        = \left(
            a_{u,i_{s_1+ \cdots +s_{u-1}+2}},
            f_{i_{s_1+\cdots +s_{u-1}+2}, i_{s_1+\cdots +s_{u-1}+3}}, \ldots,
            \right.
        \\
            \left.
            \phantom{a_{i_{s_{u-1}}}} 
            f_{i_{s_1+\cdots +s_{u}-1}, i_{s_1+\cdots +s_{u}}},
            b_{t,v_u,i_{s_1+\cdots +s_{u}}}
        \right)
        .
    \end{multline*}
    Clearly, $P_u$ corresponds to $K_u$ in the theorem. Note that the sequences
    $P_1, \ldots, P_h$ by the very definition of a permutation are edge disjoint
    in the sense that 
    \begin{enumerate}[label=(\arabic*)]
        \item $n_{i,j}$ occurs at most once in $P_1, \ldots, P_h$,
        \item if $n_{j,l_1}, n_{j,l_2}$ occur in $P_1, \ldots, P_h$ then
            $l_1=l_2$,
        \item if $n_{j_1,l}, n_{j_2,l}$ occur in $P_1, \ldots, P_h$ then
            $j_1=j_2$.
    \end{enumerate}
    Having indexed $s_1+\cdots + s_h$ of the integers in $\{1, \ldots, |E|+h\}$
    we consider what is left, namely
    \begin{equation*}
        \Lambda=\{ 1, \ldots, |E|+h\} \setminus \{i_1, \ldots, i_{s_1+\ldots
        +s_h}\}
        .
    \end{equation*}
    By construction we have $i_1=|E|+1$, $\ldots, i_{s_1+\cdots
    +s_{h-1}+1}=|E|+h$ and therefore $\Lambda \subseteq \{1, \ldots, |E|\}$.
    Also by construction for every 
    \begin{equation*}
        \delta \in \{ 1, \ldots, |E| \} \cap \{ i_1, \ldots, i_{s_1+\cdots
        +s_h} \}
    \end{equation*}
    we have $\delta=p(\epsilon)$ for some $\epsilon \in \{i_1, \ldots, i_{s_1+
    \cdots +s_h} \}$. Therefore $p(\lambda) \in \Lambda$ for all $\lambda \in
    \Lambda$ holds. In particular $p^{(x)}(\lambda) \in \{1, \ldots, |E|\}$ for
    all $x$. From Lemma~\ref{lemnocycle} we conclude that $p(\lambda)=\lambda$
    for all $\lambda \in \Lambda$.
\end{proof}

\end{document}